\documentclass[orivec,envcountsame,envcountsect,a4paper,runningheads]{llncs}

\usepackage{cite}
\usepackage{graphicx}
\usepackage{ifthen}
\usepackage{xspace}
 
\newboolean{talk}
\setboolean{talk}{false}
\newboolean{paper}
\setboolean{paper}{false}

\newboolean{IEEEstyle}
\setboolean{IEEEstyle}{false}
\newboolean{lipicsstyle}
\setboolean{lipicsstyle}{false}
\newboolean{eptcsstyle}
\setboolean{eptcsstyle}{false}
\newboolean{entcsstyle}
\setboolean{entcsstyle}{false}
\newboolean{sigplanstyle}
\setboolean{sigplanstyle}{false}
\newboolean{easychairstyle}
\setboolean{easychairstyle}{false}
\newboolean{scrartclstyle}
\setboolean{scrartclstyle}{false}
\newboolean{lncsstyle}
\setboolean{lncsstyle}{false}

\newboolean{needstheorems}
\setboolean{needstheorems}{false}
\newboolean{withimages}
\setboolean{withimages}{false}
\newboolean{withproofs}
\setboolean{withproofs}{true}

\newboolean{french}
\setboolean{french}{false}

\setboolean{lncsstyle}{true}

\usepackage{amsthm}
\usepackage{proof}
\usepackage{tikz}
\usetikzlibrary{cd,positioning}
\ifthenelse{\boolean{talk}}{}{\usepackage[utf8]{inputenc}}
\ifthenelse{\boolean{french}}{\usepackage[french]{babel}}{
  \ifthenelse{\boolean{lipicsstyle}}{}{\usepackage[english]{babel}}}
\usepackage{amssymb}
\usepackage{amsmath}
\usepackage{graphicx}
\usepackage{bussproofs}
\usepackage{cmll}

\usepackage{todonotes}

 \usepackage{multirow}
 \usepackage{fancybox}
 \usepackage{hhline}
\usepackage{nameref}
\usepackage{complexity}
\usepackage{appendix}
\usepackage{xcolor}
\usepackage[bookmarks=false]{hyperref}
\hypersetup{
	colorlinks=true,
	citecolor=blue,
	linkcolor=black,
    	urlcolor=black,
}
\usepackage{stmaryrd}
\SetSymbolFont{stmry}{bold}{U}{stmry}{m}{n}
\usepackage[normalem]{ulem}

\ifthenelse{\boolean{IEEEstyle}}{
\setboolean{needstheorems}{true}}{}

\ifthenelse{\boolean{eptcsstyle}}{
\setboolean{needstheorems}{true}}{}

\ifthenelse{\boolean{scrartclstyle}}{
\setboolean{needstheorems}{true}}{}

\ifthenelse{\boolean{sigplanstyle}}{
\setboolean{needstheorems}{true}}{}

\ifthenelse{\boolean{easychairstyle}}{
\setboolean{needstheorems}{true}}{}

\ifthenelse{\boolean{lipicsstyle}}{
    }{}

\ifthenelse{\boolean{needstheorems}}{
\input{\macrospath/ben-macros-thm-environments}}{}

\newcommand{\myproof}[1]{
\ifthenelse{\boolean{withproofs}}{#1}{}
}

\newcommand{\la}[1]{\lambda #1.}

\newcommand{\tm}{t}
\newcommand{\tmtwo}{s}
\newcommand{\tmthree}{u}
\newcommand{\tmfour}{r}
\newcommand{\tmfive}{p}

\newcommand{\var}{x}
\newcommand{\vartwo}{y}

\newcommand{\rootRew}[1]{\mapsto_{#1}}
\newcommand{\Rew}[1]{\rightarrow_{#1}}
\newcommand{\Rewn}[1]{\rightarrow_{#1}^*}

\renewcommand{\to}{\Rew{}}
\newcommand{\tob}{\Rew{\beta}}
\newcommand{\rtob}{\rootRew{\beta}}

\newcommand{\lRew}[1]{\; \mbox{}_{#1}{\leftarrow}\ }
\newcommand{\lRewn}[1]{\; \mbox{}^{*}_{#1}{\leftarrow}\ }

\newcommand{\ctxholep}[1]{\langle #1\rangle}
\newcommand{\ctxhole}{\ctxholep{\cdot}}

\makeatletter
\newsavebox{\@brx}
\newcommand{\llangle}[1][]{\savebox{\@brx}{\(\m@th{#1\langle}\)}%
  \mathopen{\copy\@brx\kern-0.7\wd\@brx\usebox{\@brx}}}
\newcommand{\rrangle}[1][]{\savebox{\@brx}{\(\m@th{#1\rangle}\)}%
  \mathclose{\copy\@brx\kern-0.7\wd\@brx\usebox{\@brx}}}
\makeatother

\newcommand{\ctx}{C}
\newcommand{\ctxtwo}{D}
\newcommand{\ctxthree}{E}

\newcommand{\ctxp}[1]{\ctx\ctxholep{#1}}

\newcommand{\nbvctxtwo}[1]{\nbvctxtwo{#1}}

\newcommand{\defeq}{:=}

\newcommand{\grameq}{::=}

\newcommand{\isub}[2]{\{#1{\shortleftarrow}#2\}}

\newcommand{\llbrace}{\{ \kern -0.27em \vert}
\newcommand{\rrbrace}{\vert \kern -0.27em \}}

\renewcommand{\l}{\lambda}
\newcommand{\ie}{{\em i.e.}\xspace}

\newcommand{\ih}{{\textit{i.h.}}\xspace}

\ifthenelse{\boolean{entcsstyle}}{}{
	}

\newcommand{\ignore}[1]{}

\newcommand{\colspace}{@{\hspace{.5cm}}}
\newcommand{\myinput}[1]{\ifthenelse{\boolean{withimages}}{\input{#1}}{}}

\newcommand{\reflemmap}[2]{Lemma~\ref{l:#1}.\ref{p:#1-#2}}

\newcommand{\refequation}[1]{(\ref{eq:#1})}

\newcommand{\set}[1]{\{#1\}}

\newcommand{\todist}{\Rew{\dist}}

\newcommand{\size}[1]{|#1|}

\newcommand{\LO}{{LO}\xspace}

\newcommand{\pair}[2]{(#1,#2)}
\renewcommand{\pair}[2]{\langle#1,#2\rangle}

\newcommand{\tolo}{\Rew{LO}}

\newcommand{\withproofs}[1]{\ifthenelse{\boolean{withproofs}}{#1}{}}

\newcommand{\withoutproofs}[1]{\ifthenelse{\boolean{withproofs}}{}{#1}}

\ifthenelse{\boolean{entcsstyle}}{}{
  \ifthenelse{\boolean{lncsstyle}}{}{	
    
  }
 }

\newcounter{numberone}

\newcommand\SN{\ensuremath{\mathsf{SN}}}
\newcommand\interp[1]{\ensuremath{\llbracket #1\rrbracket}}
\newcommand\Red[1]{\ensuremath{\mathsf{Red}(#1)}}

\newcommand{\distsym}{\mathsf{dist}}
\newcommand{\distl}{\l_{\distsym}}

\newcommand{\paapsym}{@_\times}
\newcommand{\prabsym}{\pi_\lambda}

\newcommand{\rtopri}{\rootRew{\pi_i}}
\newcommand{\rtopaap}{\rootRew{\paapsym}}
\newcommand{\rtoprab}{\rootRew{\prabsym}}

\newcommand{\topri}{\Rew{\pi_i}}

\newcommand{\toprab}{\Rew{\prabsym}}

\renewcommand{\todist}{\Rew{\mathsf{dist}}}
\renewcommand{\size}[1]{\mathsf{size}(#1)}

\newcommand{\eval}[1]{\mathsf{eval}(#1)}

\newcommand\xrecap[4]{{\noindent\sffamily\bfseries}{\bf #1 \ref{#3} (#2)}{\bf.} {\em #4}}

\begin{document}
\title{Functional Pearl: the Distributive $\l$-Calculus}

\author{Beniamino Accattoli\inst{1}
   \and
   Alejandro D\'{\i}az-Caro\inst{2,3,}
}

\authorrunning{Accattoli \& D\'{\i}az-Caro}
\institute{Inria \& LIX, École Polytechnique, UMR 7161, Palaiseau, France
  \email{beniamino.accattoli@inria.fr}
  \and
  CONICET-Universidad de Buenos Aires. Instituto de Ciencias de la
  Computaci\'on. Buenos Aires, Argentina.
  \and
  Departamento de Ciencia y Tecnolog\'{\i}a. Universidad Nacional de Quilmes,\\
  Bernal, BA, Argentina\\
  \email{adiazcaro@icc.fcen.uba.ar}}

\maketitle              

\begin{abstract}
We introduce a simple extension of the $\lambda$-calculus with pairs---called the distributive $\lambda$-calculus---obtained by adding a computational interpretation of the valid distributivity isomorphism $A \Rightarrow (B\wedge C)\ \ \equiv\ \  (A\Rightarrow B) \wedge (A\Rightarrow C)$ of simple types. 
We study the calculus both as an untyped and as a simply typed setting. Key features of the untyped calculus are confluence, the absence of clashes of constructs, that is, evaluation never gets stuck, and a leftmost-outermost normalization theorem, obtained with straightforward proofs. 
With respect to simple types, we show that the new rules satisfy subject reduction if types are considered up to the distributivity isomorphism. The main result is strong normalization for simple types up to distributivity. The proof is a smooth variation over the one for the $\l$-calculus with pairs and simple types.

\end{abstract}

\section{Introduction}
The topic of this paper is an extension of the $\l$-calculus with pairs, deemed the \emph{distributive $\l$-calculus}, obtained by adding a natural computational interpretation of the distributivity isomorphism of simple types:
\begin{equation}
A \Rightarrow (B\wedge C)\ \ \equiv\ \  (A\Rightarrow B) \wedge (A\Rightarrow C)
\label{eq:distributivity}
\end{equation}
Namely, one extends the calculus with the following commutation rules:
\[\begin{array}{c\colspace \colspace \colspace c\colspace cc}
\pair\tm\tmtwo \tmthree \to \pair{\tm\tmthree}{\tmtwo\tmthree} & \pi_i (\la\var\tm) \to \la\var \pi_i\tm& i = 1,2
\end{array}\]
The aim of this paper is showing that the distributive $\l$-calculus is a natural system, and contributions are in both the typed and untyped settings. 

We study the untyped setting to show that our calculus makes perfect sense also without types. This is to contrast with System I, another calculus providing computational interpretations of type isomorphisms recently introduced by D\'{\i}az-Caro and Dowek \cite{DBLP:conf/rta/Diaz-CaroD19}, that does not admit an untyped version---the relationship between the two is discussed below.  

\paragraph{Typing Up to Distributivity and Subject Reduction.} At the typed level, the key point is that simple types are here considered \emph{up to distributivity}. In this way, the apparently ad-hoc new rules do satisfy the subject reduction property. 

Consider for instance $\pi_1 (\la\var\tm)$: working up to the distributivity
isomorphism---so that isomorphic types type the same terms---the subterm
$\la\var\tm$ may now have both the arrow type $A \Rightarrow (B\wedge C)$ and the conjunctive type $(A\Rightarrow B) \wedge (A\Rightarrow C)$, so that $\pi_1 (\la\var\tm)$ can be typed with $A\Rightarrow B$. Distributivity also allows for the type to be preserved---that is, subject reduction holds. According to the arrow type, indeed, the body $\tm$ of the abstraction has type $B\wedge C$ and thus the reduct of the commutation rule $\pi_1 (\la\var\tm) \to \la\var \pi_1\tm$ can also be typed with $A\Rightarrow B$. The other commutation rule can be typed similarly.

\paragraph{Overview of the Paper.} For the untyped setting, we show that the distributive
$\l$-calculus is confluent, its closed normal forms are values, and it has a
leftmost-outermost normalization theorem, exactly as for the $\l$-calculus
(without pairs).

With respect to types, we show subject reduction and strong normalization of the distributive $\l$-calculus with simple types up to distributivity. 

\paragraph{The Pearl.} The proofs in the paper are remarkably smooth. The properties for the untyped calculus are immediate. Confluence follows by the fact that the calculus is an orthogonal higher-order rewriting system \cite{AczelHO,phdklop,DBLP:conf/lics/Nipkow91}. The leftmost-outermost normalization theorem, similarly, follows by an abstract result by van Ramsdonk \cite{DBLP:conf/tlca/Raamsdonk97}, because the calculus verifies two additional properties of orthogonal higher-order rewriting system from which leftmost-outermost normalization follows. Finally, the fact that closed normal forms are values---what we call \emph{progress}---is obtained via a straightforward induction.

For the typed setting, the given argument for subject reduction goes smoothly through. The main result of the paper is that the simply typed distributive
$\l$-calculus is strongly normalizing. The proof follows
Tait's reducibility method. In particular, the interpretation of types is the
same at work for the $\l$-calculus with pairs and projections (that is, without
distributive rules). The key point is to prove that the two sides of the
distributivity isomorphism have the same interpretation. This can be proved with
two easy lemmas. Everything else is as in the case without distributive rules.

\paragraph{Type Isomorphisms and System I.}
As shown by Bruce, Di Cosmo and Longo \cite{BruceDiCosmoLongoMSCS92}
the isomorphisms of simple types can be completely characterized by
distributivity (that is, equation \refequation{distributivity}) plus the following three (for more about type isomorphisms see Di Cosmo's short survey \cite{DiCosmoMSCS05} or book \cite{DiCosmo95}):
\[\begin{array}{r\colspace rll}
\mathsf{Commutativity} & A \wedge B & \equiv &  B\wedge A
\\
\mathsf{Associativity} & (A \wedge B) \wedge C & \equiv & A\wedge (B\wedge C)
\\
\mathsf{Currying} & (A \wedge B) \Rightarrow C & \equiv & A\Rightarrow (B\Rightarrow C)
\end{array}\]
At the inception of D\'{\i}az-Caro and Dowek's System I \cite{DBLP:conf/rta/Diaz-CaroD19}, there is the idea of turning all these type isomorphisms into computational
principles. Precisely, these isomorphisms give rise to some equations $\tm \sim \tmtwo$ between terms, such as $\pair\tm\tmtwo \sim \pair\tmtwo\tm$ for the commutativity of conjunctions, for instance. The result of D\'{\i}az-Caro and Dowek is that the $\l$-calculus with pairs extended with 5 such equations (distributivity induces 2 equations) is strongly normalizing modulo. 

\paragraph{System I Rests on Types.} The equations of System I, while well behaved with respect to termination, come with two drawbacks. First, the calculus is not confluent. Second, the definitions of the rewriting rules and of the equations depend on types, so that it is not possible to consider an untyped version. Both issues are easily seen considering the commutativity equation. Consider $\tm = \pi_1 \pair\tmtwo\tmthree$. If pairs are commutative, $\tm$ can rewrite to both $\tmtwo$ and $\tmthree$:
$$\tmtwo \ \lRew{}\ \pi_1 \pair\tmtwo\tmthree \sim \pi_1 \pair\tmthree\tmtwo \ \to\  \tmthree$$ 
which breaks both confluence and subject reduction (if $\tmtwo$ has type $A$ and $\tmthree$ has type $B$). To recover subject reduction, one uses a projection $\pi_A$ indexed by a type rather than a coordinate so that (if $\tmtwo$ has type $A$ and $\tmthree$ has type $B$):
$$\tmtwo \ \lRew{}\ \pi_A \pair\tmtwo\tmthree \sim \pi_A \pair\tmthree\tmtwo \ \to\  \tmtwo$$ 
note that in order to apply the rule we need to know the type of $\tmtwo$.
Moreover, confluence is not recovered---if both $\tmtwo$ and $\tmthree$ have
type $A$ then the result may non-deterministically be $\tmtwo$ or $\tmthree$,
according to System I. D\'{\i}az-Caro and Dowek in \cite{DBLP:conf/rta/Diaz-CaroD19} indeed adopt a sort of \emph{proof-irrelevant} point of view, for which subject reduction is more important than confluence for normalization: types guarantee the existence of a result (strong normalization), and this guarantee is stable by evaluation (subject reduction), while uniqueness of the result is abandoned (no confluence).

\paragraph{System I and the Distributive $\l$-Calculus.} The two issues of System I are not due only to the 
commutativity isomorphism, as the currying and associativity isomorphisms also contribute to them. The distributive 
$\l$-calculus essentially restricts System I by keeping only the distributive isomorphism, which is the only one not 
hindering confluence and the possibility of defining the calculus independently from the type system. 

To be precise, we do not simply restrict to distributivity, but we also change
its computational interpretation. First, we do not consider
equations, but rewriting rules, and also we consider the rule $\pi_i
(\la\var\tm) \to \la\var \pi_i\tm$ that was not part of System
I\footnote{Such a rule was however present in an early version of
  System I, see~\cite{DiazcaroDowekDCM13}.}, while we remove both equations:
\[\begin{array}{c\colspace \colspace \colspace c\colspace cc}
\la\var\pair\tm\tmtwo \sim \pair{\la\var\tm}{\la\var\tmtwo} & \pi_i (\tm \tmtwo) \sim \la\var (\pi_i\tm) \tmtwo& i = 1,2
\end{array}\]
The main reason is that they would make much harder to establish confluence of the calculus, because they introduce various critical pairs---the distributive $\l$-calculus is instead trivially confluent, because it is an orthogonal higher-order rewriting system, and all such systems are confluent. 

To sum up, System I aims at being a maximal enrichment of the $\l$-calculus with computation principles induced by type isomorphisms, while the distributive $\l$-calculus rather is a minimal extension aiming at being as conservative as possible with respect to the $\l$-calculus, and in particular at being definable without types.

\paragraph{Clashes.} Let us point out a pleasant by-product of the distributive rewriting rules that we adopt. A nice property of the $\l$-calculus is that there can never be \emph{clashes} of constructors. In logical terms, there is only one introduction rule (corresponding to the abstraction constructor) and only one elimination rule (application) and they are duals, that is, they interact via $\beta$-reduction. Extensions of the $\l$-calculus usually lack this property. Typically, extending the $\l$-calculus with pairs $\pair\tm\tmtwo$ (and of course projections $\pi_1\tm$ and $\pi_2 \tm$) introduces the following two clashes:
$\pair\tm\tmtwo \tmthree$ and 
$\pi_i (\la\var\tm)$, for 
$i = 1,2$,
where an elimination constructor (application or projection) is applied to the wrong introduction rule (pair or abstraction). These clashes are stuck, as there are no rules to remove them, and it is not clear whether it makes any sense to consider such an unrestricted $\l$-calculus with pairs. 

Our distributive rules deal exactly with these clashes, removing them by commuting constructors. Concretely, the absence of clashes materializes as a \emph{progress} property: all closed normal forms are values, that is, their outermost constructor corresponds to an introduction rule.

\paragraph{Related work.} Beyond D\'{\i}az-Caro and Dowek's System I, we are aware of only three works bearing some 
analogies to ours. The first one is Arbiser, Miquel, and R{\'{\i}}os' $\l$-calculus with constructors 
\cite{DBLP:journals/jfp/ArbiserMR09}, where the $\l$-calculus is extended with constructors and a pattern matching 
construct that commutes with applications. They show it to be confluent and even having a separation theorem akin to 
Bohm's. The calculus has been further studied in a typed setting by Petit \cite{DBLP:journals/corr/abs-1009-3429}, but 
type isomorphisms play no role in this case.

The second related work is A{\"\i}t-Kaci and Garrigue's label-selective
$\lambda$-calculus \cite{AITKACI1995353}, which considers the $\lambda$-calculus
plus the only type isomorphism for the implication: $A\Rightarrow B\Rightarrow
C\equiv B\Rightarrow A\Rightarrow C$\footnote{With conjunction, this isomorphism
  is a
  consequence of currying and commutativity.}. In order to avoid losing
confluence and subject reduction, they introduce a labeling system to the
arguments, so that the application order becomes irrelevant.

Last, the untyped distributive $\lambda$-calculus coincides with the extensionality-free fragment of St{\o}vring's  
$\lambda_{FP}$ \cite{Sto06}. St{\o}vring uses it as a technical tool to study confluence and conservativity of 
surjective pairing. He points out---as we do---that the calculus is confluent because it is an 
orthogonal higher-order rewriting system, but then he gives nonetheless a proof using Tait-Martin 
L\"of's technique.

\section{The Untyped Distributive \texorpdfstring{$\l$}{lambda}-Calculus}
The language of the distributive $\l$-calculus $\distl$ is given by the
following grammar:
\[\begin{array}{r\colspace rll}
  \textsf{Terms} & \tm,\tmtwo,\tmthree &\ \ \grameq\ \  &\var \mid \la\var\tm\mid \tm\tmtwo\mid \pair\tm\tmtwo \mid \pi_1 \tm \mid\pi_2 \tm
\end{array}\]
The rewriting rules are first given at top level:
\[\begin{array}{r\colspace r@{\hspace{.2cm}}l@{\hspace{.2cm}}l\colspace l}
\multicolumn{4}{c}{\textsc{Rules at top level}}
\\
\textsf{Standard rules} &   (\lambda x.\tm)\tmtwo &\rtob &\tm\isub\var\tmtwo
\\
  &
  \pi_i \pair{\tm_1}{\tm_2} & \rtopri& \tm_i & i=1,2
  \\\\
\textsf{Distributive rules} &   \pair \tm  \tmtwo \tmthree &\rtopaap &\pair {\tm\tmthree}{\tmtwo\tmthree}
\\
  &
  \pi_i(\lambda x.\tm) &\rtoprab& \lambda x.\pi_i\tm & i=1,2
  \end{array}\]
Then, we extend them to be applied wherever in a term. We formulate such an extension using  contexts, that are terms where exactly one subterm has been replaced with a hole $\ctxhole$:
\[\begin{array}{r\colspace rll}
  \textsf{Contexts} & \ctx,\ctxtwo,\ctxthree &\grameq &\ctxhole \mid \la\var\ctx \mid \ctx\tm \mid \tm\ctx \mid \pair\ctx\tm \mid \pair\tm\ctx \mid \pi_1 \ctx\mid\pi_2 \ctx
\end{array}\]
The operation of replacing the hole $\ctxhole$ of a context $\ctx$ with a given term $\tm$ is called \emph{plugging} and it is noted $\ctxp\tm$. As usual, plugging can capture variables. Now we can define the contextual closure of the top level rules.
\begin{center}
\begin{tabular}{c}
\textsc{Contextual closure}\\ 
\infer[a \in \set{\beta, \pi_1,\pi_2,@_\times,\pi_\lambda}]{\ctxp\tm \Rew{a} \ctxp\tmtwo}{\tm\rootRew{a}\tmtwo}  
\end{tabular}
\end{center}
The contextual closure is given with contexts as a compact way of expressing the closure of all rules by all constructors, in the proofs sometimes we consider the closure by a single constructor. We use $\todist$ for the union of all the rewriting rules defined above.

\paragraph{Values and Neutral Terms.} Two subsets of terms play a special role in the following, terms whose outermost constructor corresponds to a logical introduction rule (values) and elimination rule (neutral terms), plus---in both cases---variables.
\begin{definition}[Values and neutral terms] 
\hfill
\begin{itemize}
  \item \emph{Values}: a term is \emph{value} if it is either a variable $\var$, an abstraction $\la\var\tm$, or a pair $\pair\tm\tmtwo$.
  \item \emph{Neutral terms}: a term is \emph{neutral} if it is either a variable $\var$, an application $\tm\tmtwo$, or a projection $\pi_i\tm$. 
\end{itemize}
\end{definition}
Sometimes, neutral terms are also required to be normal. Here they are not.

\paragraph{Progress.} The first property that we show is that all closed normal
forms are \emph{values}. Please note that evaluation is not call-by-value, here
the aim is simply to stress that in the distributive $\l$-calculus there are no
clashes, \ie closed-normal 
neutral terms.

\begin{proposition}[Progress]\label{prop:progress}
  If $\tm$ is a closed normal form then it is a value.
\end{proposition}
\begin{proof}
  By induction on $\tm$. Cases:
  \begin{itemize}
  \item \emph{Variable}: impossible, since $\tm$ is closed.
  \item \emph{Abstraction or pair}: then the statement holds.
  \item \emph{Application}, \ie $\tm=\tmtwo\tmthree$. Since $\tm$ is normal and closed, so is $\tmtwo$. Then, by \ih $\tmtwo$ is a value, that is, either an abstraction or a
    pair. In the first case,
    rule $\beta$ applies and in the second case rule $@_\times$ applies. Hence, in any case $\tm$ is not in normal
    form, absurd. Therefore, $\tm$ cannot be an application in normal form.
  \item \emph{Projection}, \ie $\tm=\pi_i\tmtwo$. Since $\tm$ is normal and closed, so is $\tmtwo$. Then, by \ih $\tmtwo$ is a value, that is, either an abstraction or a
    pair. In the first case, rule $\pi_\lambda$ applies and in the
    second case rule $\pi_i$ applies. Therefore, $\tm$ cannot be a projection in normal form.
    \qedhere
  \end{itemize}
\end{proof}

\paragraph{Substitution.} For the proof of strong normalization we shall need a basic property of substitution with respect to rewriting steps.

\begin{lemma}[Substitutivity of $\todist$]
\label{l:substitutivity} 
\hfill
 \begin{enumerate}
  \item \label{p:substitutivity-left}
  \emph{Left substitutivity}: if $\tm \todist \tm'$ then $\tm \isub\var\tmtwo \todist \tm'\isub\var\tmtwo$.
  
  \item \label{p:substitutivity-right}
  \emph{Right substitutivity}: if $\tmtwo \todist \tmtwo'$ then $\tm \isub\var\tmtwo \todist^* \tm\isub\var{\tmtwo'}$.
 \end{enumerate}
\end{lemma}

\begin{proof}
The first point is an easy induction on the relation $\todist$, the second one on $t$. Details in the Appendix.
\end{proof}

\paragraph{Confluence.} The distributive $\l$-calculus is an example of orthogonal higher-order rewriting system \cite{AczelHO,phdklop,DBLP:conf/lics/Nipkow91}, that is a class of rewriting systems for which confluence always holds, because of the good shape of its rewriting rules.

\begin{theorem}[Confluence]
 The distributive $\l$-calculus is confluent, that is, if $\tmtwo_1 \lRewn{\distsym} \tm \Rewn{\distsym} \tmtwo_2$ then there exists $\tmthree$ such that $\tmtwo_1 \Rewn{\distsym} \tmthree \lRewn{\distsym} \tmtwo_2$.
 \qed
\end{theorem}

\paragraph{Leftmost-Outermost Normalization.} A classic property of the ordinary $\l$-calculus is the (untyped) normalization theorem for leftmost-outermost (shortened to LO) reduction. The theorem states that LO reduction $\tolo$ is \emph{normalizing}, that is, $\tolo$ reaches a normal form from $\tm$ whenever $\tm$ has a $\beta$ reduction sequence to a normal form. The definition of LO reduction $\tolo$ on ordinary $\l$-terms is given by:
 \begin{center}
 $\begin{array}{c\colspace \colspace ccccc}
 \multicolumn{2}{c}{\textsc{\LO reduction for the ordinary $\l$-calculus}}
 \\[5pt]
	\AxiomC{}
	\UnaryInfC{$(\la\var \tm) \tmtwo \tolo \tm\isub\var\tmtwo$}
	\DisplayProof
	&
	\AxiomC{$\tm \tolo \tmtwo$}
	\AxiomC{$\tm$ is neutral}
	\BinaryInfC{$\tm \tmthree \tolo \tmtwo \tmthree$}
	\DisplayProof \\\\
	
	\AxiomC{$\tm \tolo \tmtwo$}
	\UnaryInfC{$\la\var\tm \tolo \la\var\tmtwo$}
	\DisplayProof 
	&
	  \AxiomC{$\tmthree$ is neutral and normal}
	\AxiomC{$\tm \tolo \tmtwo$}
	\BinaryInfC{$\tmthree \tm \tolo \tmthree \tmtwo$}
	\DisplayProof 
\end{array}$
\end{center}
By exploiting an abstract result by van Ramsdonk, we obtain a LO normalization theorem for $\distl$ for free. Leftmost-outermost reduction $\tolo$ can indeed be defined uniformly for every orthogonal rewriting system. For the distributive $\l$-calculus we simply consider the previous rules with respect to terms in $\distl$, and add the following clauses:
   \begin{center}
 $\begin{array}{c\colspace c\colspace cccc}
 \multicolumn{3}{c}{\textsc{\LO reduction clauses for pairs and projections}}
 \\[5pt]
	\AxiomC{}
	\UnaryInfC{$\pi_i \pair{\tm_1}{\tm_2}  \tolo \tm_i$}
	\DisplayProof
	&
	\AxiomC{}
	\UnaryInfC{$\pair \tm \tmtwo \tmthree \tolo \pair {\tm\tmthree}{\tmtwo\tmthree}$}
	\DisplayProof
	&
	\AxiomC{}
	\UnaryInfC{$\pi_i(\lambda x.\tm) \tolo \lambda x.\pi_i\tm$}
	\DisplayProof
	\\\\
	\AxiomC{$\tm \tolo \tmtwo$}
	\UnaryInfC{$\pi_i\tm \tolo \pi_i\tmtwo$}
	\DisplayProof 
	&

	\AxiomC{$\tm \tolo \tmtwo$}
	\UnaryInfC{$\pair\tm \tmthree \tolo \pair\tmtwo \tmthree$}
	\DisplayProof
&	
	  \AxiomC{$\tmthree$ is normal}
	\AxiomC{$\tm \tolo \tmtwo$}
	\BinaryInfC{$\pair\tmthree \tm \tolo \pair\tmthree \tmtwo$}
	\DisplayProof 
\end{array}$
\end{center}
In \cite{DBLP:conf/tlca/Raamsdonk97}, van Ramsdonk shows that every orthogonal higher-order rewriting system that is \emph{fully extended} and \emph{left normal} has a LO normalization theorem\footnote{Precisely, on the one hand van Ramsdonk in \cite{DBLP:conf/tlca/Raamsdonk97} shows that full extendedness implies that outermost-fair strategies are normalizing. On the other hand, left-normality implies that leftmost-fair rewriting is normalizing. Then, the LO stategy is normalizing.}. These requirements, similarly to orthogonality, concern the shape of the rewriting rules---see \cite{DBLP:conf/tlca/Raamsdonk97} for exact definitions. Verifying that the distributive $\l$-calculus is fully extended and left normal is a routine check, omitted here to avoid defining formally higher-order rewriting systems. The theorem then follows.

\begin{theorem}[Leftmost-outermost normalization]
If $\tm \todist^* \tmtwo$ and $\tmtwo$ is $\todist$-normal then $\tm \tolo^* \tmtwo$.
\qed
\end{theorem}

\section{Simple Types Up To Distributivity}
In this section we define the simply typed distributive $\l$-calculus and prove subject reduction.

\paragraph{The type system.} The grammar of types is given by
\[
  A\ \ \grameq\ \ \tau\mid A\Rightarrow A\mid A\wedge A
\]
where $\tau$ is a given atomic type.

The relation $\equiv$ denoting type isomorphism is defined by
\[\begin{tabular}{c\colspace \colspace c\colspace c\colspace c\colspace ccc}
  \infer{A\equiv A}{}
  &
  \infer{A\equiv B}{B\equiv A}
  &
  \infer{A\equiv C}{A\equiv B & B \equiv C}
  &
  \multicolumn{2}{c}{
  \infer{A\Rightarrow B\wedge C  \equiv (A\Rightarrow B)\wedge(A\Rightarrow C)}{}
  }
  \\\\
  \multicolumn{2}{c}{
  \infer{A\Rightarrow B\equiv C \Rightarrow B}{A\equiv C}
  }
  &
  \infer{A\Rightarrow B\equiv A \Rightarrow C}{B\equiv C}
  &
  \infer{A\wedge B\equiv C \wedge B}{A\equiv C}
  &
  \infer{A\wedge B\equiv A \wedge C}{B\equiv C}
  \end{tabular}
\]
The typing rules are:
\[\begin{tabular}{c\colspace c\colspace c}
  \infer[(ax)]{\Gamma,x:A\vdash x:A}{}
  &
  \multicolumn{2}{c}{
  \infer[(\equiv)]{\Gamma\vdash \tm:B}{\Gamma\vdash \tm:A & A\equiv B}
  }
\\\\
\infer[(\Rightarrow_i)]{\Gamma\vdash\lambda x.\tm:A\Rightarrow B}{\Gamma,x:A\vdash \tm:B}
&
  \multicolumn{2}{c}{
  \infer[(\Rightarrow_e)]{\Gamma\vdash \tm\tmtwo:B}{\Gamma\vdash \tm:A\Rightarrow B & \Gamma\vdash \tmtwo:B}
  }
\\\\
  \infer[(\wedge_i)]{\Gamma\vdash\pair{\tm}{\tmtwo}:A\wedge B}{\Gamma\vdash \tm:A & \Gamma\vdash \tmtwo:B}
  &
  \infer[(\wedge_{e_1})]{\Gamma\vdash\pi_1\tm:A}{\Gamma\vdash \tm:A\wedge B}
  &
  \infer[(\wedge_{e_2})]{\Gamma\vdash\pi_2\tm:B}{\Gamma\vdash \tm:A\wedge B}
  \end{tabular}
\]
Note rule $\equiv$: it states that if $\tm$ is typable with $A$ then it is also typable with $B$ for any type $B \equiv A$. It is the key rule for having subject reduction for the distributive $\l$-calculus.

\paragraph{Subject reduction.} The proof of subject reduction is built in a standard way, from a generation and a substitution lemma, plus a straightforward lemma on the shape of isomorphic types.
\begin{lemma}
  [Generation]\label{lem:generation}
  Let $\Gamma\vdash \tm:A$. Then,
  \begin{enumerate}
  \item If $\tm=x$, then $\Gamma=\Gamma',x:B$ and $B\equiv A$.  
  \item If $\tm=\lambda x.s$, then $\Gamma,x:B\vdash s:C$ and $B\Rightarrow
    C\equiv A$.
  \item If $\tm=\pair{\tmtwo_1}{\tmtwo_2}$, then $\Gamma\vdash \tmtwo_i:B_i$, for $i=1,2$, and $B_1\wedge B_2\equiv A$.
  \item If $\tm=\tmtwo\tmthree$, then $\Gamma\vdash \tmtwo:B\Rightarrow A$, $\Gamma\vdash \tmthree:A$.
  \item If $\tm=\pi_i\tmtwo$, then $\Gamma\vdash \tmtwo:B_1\wedge B_2$ and $B_i = A$.
  \end{enumerate}
\end{lemma}
\begin{proof}
  Formally, the proof is by induction on $\Gamma\vdash \tm:A$, but we rather give an informal explanation. If $\tm$ is a value ($\var$, $\la\var\tmtwo$, or $\pair{\tmtwo_1}{\tmtwo_2}$) then the last rule may be either the corresponding introduction rule or $\equiv$, and the statement follows. If $\tm$ is not a value there are two similar cases. If $\tm=\tmtwo\tmthree$ what said for values still holds, but we can say something more. Note indeed that if $A\equiv C$ and $\Gamma\vdash \tmtwo:B\Rightarrow C$ then since C is a sub-formula of $B\Rightarrow C$ we can permute the $\equiv$ rule upwards and obtain $\Gamma\vdash \tmtwo:B\Rightarrow A$. Similarly if $\tm=\pi_i\tmtwo$, which is also an  elimination rule.
\end{proof}

\begin{lemma}
  [Substitution]\label{lem:substitution}
  If $\Gamma,x:A\vdash \tm:B$ and $\Gamma\vdash \tmtwo:A$, then $\Gamma\vdash \tm\isub x\tmtwo:B$.
\end{lemma}
\begin{proof}
  Easy induction on the derivation of $\Gamma,x:A\vdash\tm:B$. Details in the Appendix.
\end{proof}

\begin{lemma}[Equivalence of types]\label{lem:equiv}
  ~
  \begin{enumerate}
  \item If $A\wedge B\equiv C\wedge D$ then $A\equiv C$ and $B\equiv D$.
  \item If $A\Rightarrow B\equiv C\Rightarrow D$ then $A\equiv C$ and $B\equiv
    C$.
  \item If $A\wedge B\equiv C\Rightarrow D$ then $D\equiv D_1\wedge D_2$,
    $A\equiv C\Rightarrow D_1$ and $B\equiv C\Rightarrow D_2$.  
  \end{enumerate}
\end{lemma}
\begin{proof}
  By induction on the definition of $\equiv$.
\end{proof}

\begin{theorem}
  [Subject reduction]\label{thm:SR}
  If $\Gamma\vdash \tm:A$ and $\tm\todist \tmtwo$, then $\Gamma\vdash \tmtwo:A$.
\end{theorem}
\begin{proof}
  By induction on $\tm\todist \tmtwo$ using the generation lemma (Lemma~\ref{lem:generation}). We first deal with the cases of the rules applied at top level:
  \begin{itemize}
  \item \emph{$\beta$-rule}: $(\lambda x.\tm)\tmtwo \rtob \tm\isub x\tmtwo$.
    By generation, $\Gamma\vdash \lambda x.\tm:B\Rightarrow A$,
    $\Gamma\vdash \tmtwo:B$.
    Again by generation, $\Gamma,x:C\vdash \tm:D$, with
    $C\Rightarrow D\equiv B\Rightarrow A$, so by Lemma~\ref{lem:equiv}, $C\equiv
    B$ and $D\equiv A$.
    Then, by rule $(\equiv)$ we have $\Gamma\vdash \tmtwo:C$, and so, by the substitution lemma (Lemma~\ref{lem:substitution}) we have $\Gamma\vdash \tm\isub x\tmtwo:D$, therefore, by
    rule $(\equiv)$, $\Gamma\vdash \tm\isub x\tmtwo:A$.
  
  \item \emph{Projection}: $\pi_i\langle \tm_1,\tm_2 \rangle\rtopri \tm_i$.
    By generation, $\Gamma\vdash\langle \tm_1,\tm_2
    \rangle:B_1\wedge B_2$ with $B_i= A$.
    By generation again, $\Gamma\vdash \tm_i:C_i$ with $C_1\wedge C_2\equiv
    B_1\wedge B_2$. Therefore, by rule $(\equiv)$, $\Gamma\vdash \tm_i:A$.
    
  \item \emph{Pair-application}: $\pair\tm\tmtwo \tmthree \rtopaap \pair{\tm\tmthree}{\tmtwo\tmthree}$.
    By generation, $\Gamma\vdash\pair\tm\tmtwo:B\Rightarrow
    A$ and $\Gamma\vdash \tmthree:B$.
    By generation again, $\Gamma\vdash \tm:C$ and $\Gamma\vdash
    \tmtwo:D$ with $C\wedge D\equiv B\Rightarrow A$.
    By Lemma~\ref{lem:equiv}, $A\equiv A_1\wedge A_2$, $C\equiv B\Rightarrow
    A_1$ and $D\equiv B\Rightarrow A_2$. Then,
   {\small \[
      \infer[(\equiv)]{\Gamma\vdash\pair{\tm\tmthree}{\tmtwo\tmthree}:A}
      {
        \infer[(\wedge_i)]{\Gamma\vdash\pair{\tm\tmthree}{\tmtwo\tmthree}:A_1\wedge A_2}
        {
          \infer[(\Rightarrow_e)]{\Gamma\vdash \tm\tmthree:A_1}{\infer[(\equiv)]{\Gamma\vdash \tm:B\Rightarrow
              A_1}{\Gamma\vdash \tm:C} & \Gamma\vdash \tmthree:B}
          &
          \infer[(\Rightarrow_e)]{\Gamma\vdash \tmtwo\tmthree:A_2}{\infer[(\equiv)]{\Gamma\vdash \tmtwo:B\Rightarrow
              A_2}{\Gamma\vdash \tmtwo:D} & \Gamma\vdash \tmthree:B}
        }
      }
    \]}
  
  \item \emph{Projection-abstraction}: $\pi_i(\lambda x.\tm) \rtoprab \lambda x.\pi_i\tm$.
    By generation, $\Gamma\vdash\lambda x.\tm:B_1\wedge B_2$ with
    $B_i = A$.
    By generation again, $\Gamma,x:C\vdash \tm:D$, with
    $C\Rightarrow D\equiv B_1\wedge B_2$. Then, by Lemma~\ref{lem:equiv},
    $D\equiv D_1\wedge D_2$, $B_1\equiv C\Rightarrow D_1$, and $B_2\equiv
    C\Rightarrow D_2$. Then, $A= C\Rightarrow D_i$, and so,
    \[      
        \infer[(\Rightarrow_i)]{\Gamma\vdash\lambda x.\pi_i\tm:C\Rightarrow D_i}
        {
          \infer[(\wedge_{e_i})]{\Gamma,x:C\vdash \pi_i\tm:D_i}
          {
            \infer[(\equiv)]{\Gamma,x:C\vdash \tm:D_1\wedge D_2}
            {\Gamma,x:C\vdash \tm:D}
          }
        }      
    \]  
    \end{itemize}
  The inductive cases are all straightforward. We give one of them, the others are along the same lines. Let $\lambda x.\tm\todist\lambda x.\tmtwo$ because $\tm\todist \tmtwo$. 
    By generation, $\Gamma,x:B\vdash \tm:C$, with $B\Rightarrow
    C\equiv A$.
    By \ih, $\Gamma,x:B\vdash \tmtwo:C$, so, by rules
    $(\Rightarrow_i)$ and $(\equiv)$, $\Gamma\vdash\lambda x.\tmtwo:A$.
\end{proof}

\section{Strong normalisation}
Here we prove strong normalization using Tait's reducibility technique. The key point shall be proving that the interpretation of types is stable by distributivity.

\begin{definition}[Basic definitions and notations] 
\hfill
\begin{itemize}
 \item \emph{SN terms}: we write \SN\ for the set of strongly normalising terms.

 \item \emph{One-step reducts}: the set $\set{\tmtwo\mid \tm\todist \tmtwo}$ of all the one-step reducts of a term $\tm$ is noted $\Red \tm$. 

 \item \emph{Evaluation length}: $\eval \tm$ is the length of the longest path starting from $\tm$ to arrive to
a normal form 
\item \emph{Size}: $\size{\tm}$ 
  is the size of the term
  $\tm$
  defined in the usual way.
\end{itemize}
\end{definition}

\paragraph{The interpretation of types.} The starting point of the reducibility technique is the definition of the interpretation of types, which is the standard one.
\begin{definition}
  [Interpretation of types]
\[  \begin{array}{r@{\hspace{.2cm}} c@{\hspace{.2cm}} l}
    \interp\tau &\defeq&\SN\\
    \interp{A\Rightarrow B}&\defeq&\{\tm\mid\forall \tmtwo\in\interp A, \tm\tmtwo\in\interp B\}\\
    \interp{A\wedge B}&\defeq&\{\tm\mid\pi_1\tm\in\interp A\textrm{ and }\pi_2\tm\in\interp B\}    
  \end{array}\]
\end{definition}

\paragraph{The reducibility properties.} The next step is to prove the standard three properties of reducibility. The proof is standard, that is, the distributive rules do not play a role here.

\begin{lemma}[Properties of the interpretation]
  For any type $A$ the following properties of its interpretation are valid.
  \begin{description}
  \item[CR1] $\interp A\subseteq\SN$.
  \item[CR2] If $\tm\in\interp A$ and $\tm\todist \tmtwo$, then $\tmtwo\in\interp A$.
  \item[CR3] If $\tm$ is neutral and $\Red \tm\subseteq\interp A$, then $\tm\in\interp A$.
  \end{description}
\end{lemma}
\begin{proof}
  ~
  \begin{description}
  \item[CR1] By induction on $A$. Cases:
    \begin{itemize}
    \item $\interp\tau=\SN$.
    \item Let $\tm\in\interp{A\Rightarrow B}$. Then, for all $\tmtwo\in\interp A$, we
      have $\tm\tmtwo\in\interp B$. By \ih, $\interp
      B\subseteq\SN$, so $\tm\tmtwo\in\SN$, and hence, $\tm\in\SN$.
    \item Let $\tm\in\interp{A\wedge B}$. Then, in particular, $\pi_1\tm\in\interp A$. By \ih, $\interp A\subseteq\SN$, so $\pi_1\tm\in\SN$, and
      hence, $\tm\in\SN$.    
    \end{itemize}
  \item[CR2] By induction on $A$. Cases:
    \begin{itemize}
    \item Let $\tm\in\interp\tau=\SN$. Then if $\tm\todist \tmtwo$, we have
      $\tmtwo\in\SN=\interp\tau$.
    \item Let $\tm\in\interp{A\Rightarrow B}$. Then, for all $\tmthree\in\interp A$, we
      have $\tm\tmthree\in\interp B$. By \ih on $B$, since $\tm\tmthree\todist
      \tmtwo\tmthree$, we have $\tmtwo\tmthree\in\interp B$ and so $\tmtwo\in\interp{A\Rightarrow B}$.
    \item Let $\tm\in\interp{A_1\wedge A_2}$. Then, $\pi_i\tm\in\interp{A_i}$, for
      $i=1,2$.
      By \ih on $A_i$, since $\pi_i\tm\todist\pi_i\tmtwo$, we have $\pi_i\tmtwo\in\interp{A_i}$ and so $\tmtwo\in\interp{A_1\wedge A_2}$.    
    \end{itemize}
  \item[CR3]  By induction on $A$. Let $\tm$ be neutral. Cases:
    \begin{itemize}
    \item Let $\Red \tm\subseteq\interp\tau=\SN$. Then $\tm\in\SN=\interp\tau$.
    \item Let $\Red \tm\subseteq\interp{A\Rightarrow B}$. Then for each $\tm'\in\Red
      \tm$, we have that for all $\tmtwo\in\interp A$, $\tm'\tmtwo\in\interp B$. Since $\tm\tmtwo$ is neutral, if we show that $\Red{\tm\tmtwo}\subseteq\interp B$ then the \ih on $B$ gives $\tm\tmtwo\in\interp B$ and so $\tm\in\interp{A\Rightarrow B}$. 
      
      Since, by CR1 on $\interp A$, we have $\tmtwo\in\SN$, we show that
      $\Red{\tm\tmtwo}\subseteq\interp B$  by a second induction on $\size \tmtwo$. The
      possible reducts of $\tm\tmtwo$ are:
      \begin{itemize}
      \item $\tm'\tmtwo$, with $\tm\todist \tm'$, which is in $\interp B$ by hypothesis,
      \item $\tm\tmtwo'$, with $\tmtwo\todist \tmtwo'$, then by the second induction
      hypothesis $\Red{\tm\tmtwo'}\subseteq\interp B$ and by \ih $\tm\tmtwo' \in \interp B$.
      \end{itemize}
       Note that since $\tm$ is neutral there are no other reductions from $\tm\tmtwo$.
    \item Let $\Red \tm\subseteq\interp{A_1\wedge A_2}$. Then for each $\tm'\in\Red \tm$,
      we have that $\pi_i\tm'\in\interp{A_i}$, for $i=1,2$. We show
      that $\Red{\pi_i\tm}\subseteq\interp{A_i}$, which---since
      $\pi_i\tm$ is neutral---by \ih implies
      $\pi_i\tm\in\interp{A_i}$, and so $\tm\in\interp{A_1\wedge A_2}$.

      Since $\tm$ is neutral, its only possible reducts have the form
      $\pi_i\tm'$, with $\tm\todist \tm'$, which are in $\interp{A_i}$ by hypothesis. 
       \qedhere
    \end{itemize}
  \end{description}
\end{proof}

\paragraph{Stability of the interpretation by isomorphism.} Finally, we come to the point where distributivity plays a role. Here we prove that the interpretation of types is stable by $\equiv$, that is, if $A\equiv B$ then $\interp A = \interp B$. We need an auxiliary lemma stating a sort of stability by anti-reduction of $\interp A$ with respect to the standard rewriting rules of $\beta$ and projection.

\begin{lemma}
\label{l:stability-aux} 
\hfill
 \begin{enumerate}  
  \item \label{p:stability-aux-three}
  If $\tm, \tmtwo\in \SN$ and $\tm\isub\var\tmtwo \in \interp A$ then $(\la\var\tm)\tmtwo \in \interp A$.
  
  \item \label{p:stability-aux-five}
  If $\tm_i \in \interp{A_i}$ then
  $\pi_i\pair{\tm_1}{\tm_2} \in \interp{A_i}$, for $i=1,2$.
 \end{enumerate}
\end{lemma}

\begin{proof}
 \hfill
 \begin{enumerate}  
  \item By induction on $\eval\tm + \eval\tmtwo$. We show that $\Red{(\la\var\tm)\tmtwo}\subseteq\interp{A}$, and obtain the statement by CR3. Cases:
  \begin{itemize}
   \item $(\la\var\tm)\tmtwo \todist (\la\var\tm')\tmtwo$ with $\tm \todist \tm'$. We can apply the \ih because if $\tm\todist \tm'$ then $\tm\isub\var\tmtwo \todist \tm'\isub\var\tmtwo$ by left substitutivity of $\todist$ (\reflemmap{substitutivity}{left}), and $\tm'\isub\var\tmtwo\in \interp A$ by CR2. By \ih, $(\la\var\tm')\tmtwo \in \interp A$.
   \item $(\la\var\tm)\tmtwo \todist (\la\var\tm)\tmtwo'$ with $\tmtwo \todist \tmtwo'$. We can apply the \ih because if $\tmtwo\todist \tmtwo'$ then $\tm\isub\var\tmtwo \todist^* \tm\isub\var{\tmtwo'}$ by right substitutivity of $\todist$ (\reflemmap{substitutivity}{right}), and $\tm\isub\var{\tmtwo'}\in\interp A$ by CR2. By \ih, $(\la\var\tm)\tmtwo' \in \interp A$.
   \item $(\la\var\tm)\tmtwo \tob \tm\isub\var\tmtwo$, which is in $\interp A$ by hypothesis.
  \end{itemize}

    \item By CR1 we have $\tm_i\in \SN$. By induction on $\eval{\tm_1} + \eval{\tm_2}$. The possible reducts of $\pi_i\pair{\tm_1}{\tm_2}$ are:
      \begin{itemize}
      \item $\tm_i$, because of a $\topri$ step. Then $\tm_i \in \interp{A_i}$ by hypothesis.
      \item $\pi_i\pair{\tm_1'}{\tm_2}$, with $\tm_1\todist \tm_1'$. We can apply the \ih because $\interp{A_1} \ni \tm_1 \todist \tm_1'$ which is in $\interp{A_1}$ by CR2. Then $\pi_i\pair{\tm_1'}{\tm_2} \in \interp{A_1}$ by \ih      
      \item $\pi_i\pair{\tm_1}{\tm_2'}$, with $\tm_2\todist \tm_2'$. As the previous case, just switching coordinate of the pair.      
        \qedhere
      \end{itemize}

  \end{enumerate}
\end{proof}

\begin{lemma}[Stability by isomorphism]\label{lem:stability}
  If $A\equiv B$, then $\interp A=\interp B$.
\end{lemma}

\begin{proof}
  By induction on $A\equiv B$. The only interesting case is the base case $A\Rightarrow
  B_1\wedge B_2\equiv(A\Rightarrow B_1)\wedge(A\Rightarrow B_2)$. The inductive 
  cases follow immediately from the \ih
  
  We prove $\interp{A\Rightarrow B_1\wedge B_2}=\interp{(A\Rightarrow
    B_1)\wedge(A\Rightarrow B_2)}$ by proving the double inclusion.
  \begin{itemize}
  \item Let $\tm\in\interp{A\Rightarrow B_1\wedge B_2}$.
    Then for all $\tmtwo\in\interp A$ we have $\tm\tmtwo\in\interp{B_1\wedge B_2}$, so
    \begin{equation}
      \label{eq:pirs}
    \pi_i(\tm\tmtwo)\in\interp{B_i}
    \end{equation}
    We need to prove that $(\pi_i\tm)\tmtwo\in\interp{B_i}$. Since this term is
    neutral, we prove that $\Red{(\pi_i\tm)\tmtwo}\subseteq\interp{B_i}$ and conclude
    by CR3.
    By CR1 and \refequation{pirs}, $\tm$ and $\tmtwo$ are in $\SN$, so we proceed by induction on $\eval \tm + \eval \tmtwo$.
    The possible one-step reducts fired from $(\pi_i\tm)\tmtwo$ are:
    \begin{itemize}
    \item $(\pi_i\tm')\tmtwo$, with $\tm\todist \tm'$, then \ih applies.
    \item $(\pi_i\tm)\tmtwo'$, with $\tmtwo\todist \tmtwo'$, then \ih applies.    
    \item $\tm_i\tmtwo$, if $\tm=\pair{\tm_1}{\tm_2}$. Since
      $\pi_i(\tm\tmtwo)=\pi_i(\pair{\tm_1}{\tm_2} \tmtwo)\todist\pi_i\pair{\tm_1\tmtwo}{\tm_2\tmtwo}\todist \tm_1\tmtwo$, by \eqref{eq:pirs} and CR2 we have
      $\tm_i\tmtwo\in\interp{B_i}$.
    \item $(\lambda x.\pi_i\tmthree)\tmtwo$ if $\tm=\lambda x.\tmthree$. Then we can apply \reflemmap{stability-aux}{three}, since we know that $\tmthree$ and $\tmtwo$ are $\SN$ and that $\pi_i (\tm \tmtwo) = \pi_i ((\la\var \tmthree) \tmtwo) \tob \pi_i \tmthree\isub\var\tmtwo$ which by \eqref{eq:pirs} and CR2 is in $\interp{B_i}$. We obtain $(\lambda x.\pi_i\tmthree)s\in \interp{B_i}$
    \end{itemize}

  \item Let $\tm\in\interp{(A\Rightarrow B_1)\wedge(A\Rightarrow B_2)}$.
    Then $\pi_i\tm\in\interp{A\Rightarrow B_i}$, and so for all $\tmtwo\in\interp A$,
    we have $(\pi_i\tm)\tmtwo\in\interp{B_i}$. By CR1 we have $\tm,\tmtwo\in\SN$, so we
    proceed by induction on $\eval \tm + \eval \tmtwo$ to show that
    $\Red{\pi_i(\tm\tmtwo)}\subseteq\interp{B_i}$, which implies
    $\pi_i(\tm\tmtwo)\in\interp{B_i}$ and so $\tm\tmtwo\in\interp{B_1\wedge B_2}$, and then
    $\tm\in\interp{A\Rightarrow B_1\wedge B_2}$. The possible reducts of 
    $\pi_i(\tm\tmtwo)$ are:
    \begin{itemize}
    \item $\pi_i(\tm'\tmtwo)$ with $\tm\todist \tm'$, then the \ih applies.
    \item $\pi_i(\tm\tmtwo')$ with $\tmtwo\todist \tmtwo'$, then the \ih applies.
    \item $\pi_i(\tmthree\isub x\tmtwo)$ if $\tm=\lambda x.\tmthree$. Then since
      $(\pi_i\tm)\tmtwo\in\interp{B_i}$, we have $(\pi_i\lambda x.\tmthree)\tmtwo\in\interp{B_i}$ and
      $(\pi_i\lambda x.\tmthree)\tmtwo\toprab(\lambda x.\pi_i\tmthree)\tmtwo\tob\pi_i(\tmthree\isub x\tmtwo)$, so, by
      CR2, $\pi_i(\tmthree\isub x\tmtwo)\in\interp{B_i}$.
    \item $\pi_i\pair{\tm_1\tmtwo}{\tm_2\tmtwo}$ if $\tm=\pair{\tm_1}{\tm_2}$. We apply \reflemmap{stability-aux}{five}, since we have
        $(\pi_i\pair{\tm_1}{\tm_2}) \tmtwo\in\interp{B_i}$ and
        $(\pi_i\pair{\tm_1}{\tm_2}) \tmtwo\topri \tm_i\tmtwo$, so, by CR2, $\tm_i\tmtwo\in\interp{B_i}$. We then obtain $\pi_i\pair{\tm_1\tmtwo}{\tm_2\tmtwo}\in \interp{B_i}$.
        \qedhere
    \end{itemize}
  \end{itemize}
\end{proof}

\paragraph{Adequacy.} The last step is to prove what is usually called \emph{adequacy}, that i\tmtwo, that typability of $\tm$ with $A$ implies that $\tm\in \interp A$, up to a substitution $\theta$ playing the role of the typing context $\Gamma$. The proof is standard, the distributive rules do not play any role.

\begin{definition}[Valid substitution]
  We say that a substitution $\theta$ is valid with respect to a context
  $\Gamma$ (notation $\theta\vDash\Gamma$) if for all $x:A\in\Gamma$, we have
  $\theta x\in\interp A$.
\end{definition}

 \begin{lemma}
   [Adequacy]\label{lem:adequacy}
   If $\Gamma\vdash \tm:A$ and $\theta\vDash\Gamma$, then $\theta \tm\in\interp A$.
 \end{lemma}
 \begin{proof}
   By induction on the derivation of $\Gamma\vdash \tm:A$.
   \begin{itemize}
     \setlength\itemsep{1mm}
   \item $\vcenter{\infer[(ax)]{\Gamma,x:A\vdash x:A}{}}$
     \quad
     Since $\theta\vDash\Gamma,x:A$,
     we have $\theta x\in\interp A$.
   \item $\vcenter{\infer[(\Rightarrow_i)]{\Gamma\vdash\lambda x.\tm:A\Rightarrow
       B}{\Gamma,x:A\vdash \tm:B}}$\\

     By \ih, if $\theta'\vDash\Gamma,x:A$, then
     $\theta'\tm\in\interp B$. Let $\tmtwo\in\interp A$, we have to prove that
     $\theta(\lambda x.\tm)\tmtwo=(\lambda x.\theta \tm)\tmtwo\in\interp B$. By CR1,
     $\tmtwo,\theta \tm\in\SN$, so we proceed by a second induction on 
     $\size \tmtwo+\size{\theta \tm}$ to show
     that $\Red{(\lambda x.\theta \tm)\tmtwo}\subseteq\interp B$, which implies
     $(\lambda x.\theta \tm)\tmtwo\in\interp B$. The possible reducts of $(\lambda
     x.\theta \tm)\tmtwo$ are:
     \begin{itemize}
     \item $(\lambda x.\tm')\tmtwo$, with $\theta \tm\todist \tm'$, then the second \ih applies.
     \item $(\lambda x.\theta \tm)\tmtwo'$, with $\tmtwo\todist \tmtwo'$, then the second \ih applies.
     \item $\theta \tm\isub xs$, then take $\theta'=\theta,x\mapsto s$ and notice that
       $\theta'\vDash\Gamma,x:A$, so $\theta \tm\isub x\tmtwo\in\interp B$.
     \end{itemize}
   \item $\vcenter{\infer[(\Rightarrow_e)]{\Gamma\vdash \tm\tmtwo:B}{\Gamma\vdash
         \tm:A\Rightarrow B & \Gamma\vdash \tmtwo:B}}$\\

     By \ih, $\theta \tm\in\interp{A\Rightarrow B}$ and
     $\theta \tmtwo\in\interp B$, so, by definition, $\theta \tm\theta
     s=\theta(\tm\tmtwo)\in\interp B$.

   \item $\vcenter{\infer[(\wedge_i)]{\Gamma\vdash\pair{\tm_1}{\tm_2} :A_1\wedge A_2}{\Gamma\vdash \tm_1:A_1 & \Gamma\vdash \tm_2:A_2}}$\\

     By \ih, $\theta \tm_i\in\interp{A_i}$, for $i=1,2$. By
     CR1 we have $\theta \tm_i\in\SN$, hence we proceed by a second induction on $\size{\theta \tm_1}+\size{\theta \tm_2}$ to show that $\Red{\pi_i\pair{\theta \tm_1}{\theta \tm_2}}\subseteq\interp{A_1}$, which, by
  CR3 implies $\pi_i\pair{\theta \tm_1}{\theta \tm_2}\in\interp{A_i}$ and so $\pair{\theta \tm_1}{\theta \tm_2}\in\interp{A_1\wedge A_2}$.

  The possible one-step reducts of $\pi_i\pair{\theta \tm_1}{\theta \tm_2}$ are:
  \begin{itemize}
  \item $\pi_i\pair{\tm'}{\theta \tm_2}$, with $\theta \tm_1\todist \tm'$, then the second \ih applies.
  \item $\pi_i\pair{\theta \tm_1}{\tm'}$, with $\theta \tm_2\todist \tm'$, then the second \ih applies.
  \item $\theta \tm_i\in\interp{A_i}$.
  \end{itemize}
     
   \item $\vcenter{\infer[(\wedge_{e_i})]{\Gamma\vdash\pi_i\tm:A_i}{\Gamma\vdash \tm:A_1\wedge A_2}}$\quad
     By \ih, $\theta \tm\in\interp{A_1\wedge A_2}$, so, by definition, $\pi_i(\theta \tm)=\theta\pi_i\tm\in\interp{A_i}$.

   \item $\vcenter{\infer[(\equiv)]{\Gamma\vdash \tm:B}{\Gamma\vdash \tm:A & A\equiv
         B}}$\quad
     By \ih, $\theta \tm\in\interp A$, so, by Lemma~\ref{lem:stability}, $\theta \tm\in\interp B$.
     \qedhere
   \end{itemize}
 \end{proof}

\begin{theorem}[Strong normalisation]
  \label{thm:SN}
  If $\Gamma\vdash \tm:A$, then $\tm\in\SN$.
\end{theorem}
\begin{proof}
  By Lemma~\ref{lem:adequacy}, if $\theta\vDash\Gamma$, $\theta \tm\in\interp A$.
  By CR3, variables---which are neutral terms---are in all the interpretations, and so the identity
  substitution is valid in any context, in particular, in $\Gamma$. Hence, $\tm\in\interp
  A$. By CR1, $\interp A\subseteq\SN$. Hence, $\tm\in\SN$.
\end{proof}

\section{Discussion and conclusions}
\label{sect:variations}

\paragraph{The Unit Type.} The point of the paper is the fact that the distributive rewriting rules and typing  up to distributivity perfectly marry together. The elimination of clashes, on the other hand, is a nice consequence of our approach that should not be taken too seriously, because it does not scale up, as we now show. 

Let's consider the extension of the distributive $\l$-calculus with the unit type $\top$ and a construct $\star$ of type $\top$. In this extended setting it is still possible to interpret distributivity as in the previous sections, and all our results still holds. There are however two new clashes, namely $\star\, \tmthree$ and $\pi_i \star$. If one makes the further step of eliminating them via new rules and type them up to new isomorphisms, then unfortunately normalization breaks, as we now show. 

Consider their natural commutation rules:
\[\begin{array}{c\colspace \colspace \colspace c\colspace cc}
\star\, \tmthree \to \star & \pi_i \star \to \star & i = 1,2
\end{array}\]
To have subject reduction along the same lines of what we did, one needs to work up to the following two isomorphisms:
\[\begin{array}{c\colspace \colspace \colspace c\colspace cc}
A \Rightarrow \top \ \equiv \ \top& \top \wedge \top \ \equiv \ \top
\end{array}\]
Note that $A \Rightarrow \top  \equiv  \top$ has to be valid for any type $A$, therefore in
particular it is true for $\top$, giving $\top \Rightarrow \top \equiv
\top$. Now, unfortunately, one can type the diverging term
$\Omega \defeq (\la\var\var\var) (\la\var\var\var)$, as the following
derivation shows, and in fact all the terms of the ordinary $\l$-calculus---said
differently strong normalization breaks.
\[
\scalebox{0.9}{
  \infer[(\Rightarrow_e)]{\vdash(\lambda x.xx)(\lambda x.xx):\top}
  {
      \infer[(\Rightarrow_i)]{\vdash\lambda x.x:\top\Rightarrow\top}
      {
        \infer[(\Rightarrow_e)]{x:\top\vdash xx:\top}
        {
          \infer[(\equiv)]{x:\top\vdash x:\top\Rightarrow\top}{\infer[(ax)]{x:\top\vdash x:\top}{}}
          &\infer[(ax)]{x:\top\vdash x:\top}{}
        }
      }
    &
    \infer[(\equiv)]{\vdash\lambda x.xx:\top}
    {
      \infer[(\Rightarrow_i)]{\vdash\lambda x.x:\top\Rightarrow\top}
      {
        \infer[(\Rightarrow_e)]{x:\top\vdash xx:\top}
        {
          \infer[(\equiv)]{x:\top\vdash x:\top\Rightarrow\top}{\infer[(ax)]{x:\top\vdash x:\top}{}}
          &\infer[(ax)]{x:\top\vdash x:\top}{}
        }
      }
    }
  }
}
\]
This example also reinforces the fact, already stressed in the introduction, that interpretations of type isomorphisms tend to break key properties. Distributivity, instead, is somewhat special, as it admits an interpretation that is conservative with respect to the properties of the underlying calculus.

\paragraph{Additional Distributivity Rules.} It is possible to add the two following distributive rewriting rules:
\[\begin{array}{c\colspace \colspace \colspace c\colspace cc}
\la\var\pair\tm\tmtwo \to \pair{\la\var\tm}{\la\var\tmtwo} & \pi_i (\tm \tmtwo) \to (\pi_i\tm) \tmtwo& i = 1,2
\end{array}\]
Subject reduction and strong normalization still hold. The problem is that the rewriting system is no longer orthogonal, since the following critical pairs are now possible:
  \begin{center}
    \begin{tikzpicture}
      \node at (0,0) {
    \begin{tikzcd}[column sep=5mm]
       \pi_i(\lambda x.\langle \tm_1,\tm_2 \rangle) \ar[d]\ar[r]&  \pi_i\langle \lambda x.\tm_1,\lambda x.\tm_2 \rangle\ar[d,dashed]\\
       \lambda x.\pi_i\langle \tm_1,\tm_2 \rangle\ar[r,dashed] & \lambda
       x.\tm_i\\
      \pi_i(\langle \tm_1,\tm_2 \rangle \tmtwo)\ar[d]\ar[r] & (\pi_i\langle \tm_1,\tm_2 \rangle)\tmtwo\ar[d,dashed]\\
      \pi_i\langle \tm_1\tmtwo,\tm_2\tmtwo \rangle\ar[r,dashed] & \tm_i\tmtwo
    \end{tikzcd}};
      \node at (6.5,0) {
    \begin{tikzcd}[column sep=5mm]
      (\lambda x.\langle \tm,\tmtwo \rangle)\tmthree\ar[d]\ar[r] & \langle \tm\isub\var\tmthree,\tmtwo\isub\var\tmthree \rangle \\
      \langle \lambda x.\tm,\lambda x.\tmtwo \rangle \tmthree\ar[r,dashed] &
       \langle (\lambda x.\tm)\tmthree,(\lambda x.\tmtwo)\tmthree \rangle\ar[u,dashed,near end,swap,"2"]\\
      \pi_i((\lambda x.\tm) \tmtwo)\ar[d]\ar[r] & \pi_i(\tm\isub\var\tmtwo)\\
      (\pi_i(\lambda x.\tm))\tmtwo\ar[r,dashed] & (\lambda x.\pi_i\tm)\tmtwo\ar[u,dashed]&
    \end{tikzcd}};
    \end{tikzpicture}
  \end{center}
While the pairs on the left side are easy to deal with, those on the right side have an unpleasant closing diagram and make the rewriting system much harder to study.

\paragraph{Conclusions.}
We have extended the $\l$-calculus with pairs with two additional commutation rules inspired by the distributivity isomorphism of simple types, and showed that it is a well behaved setting. In the untyped case, confluence, progress, and leftmost-outermost normalization are obtained essentially for free. In the typed case, subject reduction up to distributivity holds, as well as strong normalization. The proof of strong normalization, in particular, is a smooth adaptation of Tait's standard reducibility proof for the $\l$-calculus with pairs.

\bigskip

\noindent\textbf{Acknowledgements.}
This work has been partially funded by the ANR JCJC grant COCA HOLA ANR-16-CE40-004-01, the ECOS-Sud grant QuCa A17C03, and the French-Argentinian International Research Project SINFIN.

\newpage
\appendix
\section{Proofs Appendix}

\xrecap{Lemma}{Substitutivity of $\todist$}{l:substitutivity} 
{\ 
 \begin{enumerate}
  \item \emph{Left substitutivity}: if $\tm \todist \tm'$ then $\tm \isub\var\tmtwo \todist \tm'\isub\var\tmtwo$.
  
  \item \emph{Right substitutivity}: if $\tmtwo \todist \tmtwo'$ then $\tm \isub\var\tmtwo \todist^* \tm\isub\var{\tmtwo'}$.
 \end{enumerate}
 }

\begin{proof}\hfill
  \begin{enumerate}
  \item By induction on the relation $\todist$. Base cases:
    \begin{itemize}
    \item Let $\tm=(\la\vartwo\tmthree)\tmfour \rtob \tmthree\isub\vartwo\tmfour =
      \tm'$.
      Then,
      \begin{align*}
        \tm\isub\var\tmtwo
        &= ((\la\vartwo\tmthree)\tmfour)\isub\var\tmtwo
          =(\la\vartwo\tmthree\isub\var\tmtwo)\tmfour\isub\var\tmtwo\\
        &\rtob (\tmthree\isub\var\tmtwo)\isub\vartwo{\tmfour\isub\var\tmtwo}
          = (\tmthree\isub\vartwo\tmfour)\isub\var\tmtwo
          =\tm'\isub\var\tmtwo
      \end{align*}
    \item Let $\tm=\pi_i \pair{\tmthree_1}{\tmthree_2} \rtopri \tmthree_i = \tm'$. Then,
       \begin{align*}
         \tm\isub\var\tmtwo
         &= (\pi_i \pair{\tmthree_1}{\tmthree_2})\isub\var\tmtwo
           = \pi_i\pair{\tmthree_1\isub\var\tmtwo}{\tmthree_2\isub\var\tmtwo}\\
         &\rtopri \tmthree_i\isub\var\tmtwo
           = \tm'\isub\var\tmtwo
       \end{align*}
    \item Let $\tm = \pair{\tmthree}{\tmfour}\tmfive
      \rtopaap\pair{\tmthree\tmfive}{\tmfour\tmfive}=\tm'$. Then,
      \begin{align*}
        \tm\isub\var\tmtwo
        &=(\pair{\tmthree}{\tmfour}\tmfive)\isub\var\tmtwo
          =\pair{\tmthree\isub\var\tmtwo}{\tmfour\isub\var\tmtwo}(\tmfive\isub\var\tmtwo)\\
        &\rtopaap \pair{\tmthree\isub\var\tmtwo\tmfive\isub\var\tmtwo}{\tmfour\isub\var\tmtwo\tmfive\isub\var\tmtwo}
          =\pair{\tmthree\tmfive}{\tmfour\tmfive}\isub\var\tmtwo
          =\tm'\isub\var\tmtwo
      \end{align*}
    \item Let $\tm = \pi_i(\la\vartwo\tmthree) \rtoprab
      \la\vartwo\pi_i\tmthree=\tm'$, Then,
      \begin{align*}
        \tm\isub\var\tmtwo
        &=\pi(\la\vartwo\tmthree)\isub\var\tmtwo
          =\pi(\la\vartwo\tmthree\isub\var\tmtwo)\\
        &\rtoprab\la\vartwo\pi_i(\tmthree\isub\var\tmtwo)
          =(\la\vartwo\pi_i\tmthree)\isub\var\tmtwo
          =\tm'\isub\var\tmtwo
      \end{align*}
    \end{itemize}
    
    We treat the inductive cases compactly via contexts. First note that a straightforward induction on $\ctx$ shows that $\ctxp\tm\isub\var\tmtwo = \ctx\isub\var\tmtwo\ctxholep{\tm\isub\var\tmtwo}$, where the substitution $\ctx\isub\var\tmtwo$ on contexts is defined as expected. Now, consider $\tm=\ctxp\tmthree\Rew{a}\ctxp\tmfour=\tm'$ with
      $\tmthree\rootRew{a}\tmfour$, for some $a\in \set{\beta,
        @_\times,\pi_1,\pi_2,\pi_\lambda}$.
      By \ih,
      $\tmthree\isub\var\tmtwo\rootRew{a}\tmfour\isub\var\tmtwo$. Hence,
      \begin{align*}
        \tm\isub\var\tmtwo
        &=\ctxp\tmthree\isub\var\tmtwo
          =\ctx\isub\var\tmtwo\ctxholep{\tmthree\isub\var\tmtwo}\\
        &\Rew{a}\ctx\isub\var\tmtwo\ctxholep{\tmfour\isub\var\tmtwo}
          =\ctxp\tmfour\isub\var\tmtwo
          =\tm'\isub\var\tmtwo
      \end{align*}

  \item By induction on $\tm$.
    \begin{itemize}
    \item Let $\tm=\var$. Then,
      $$\tm\isub\var\tmtwo=\tmtwo \todist\tmtwo'=\tm\isub\var{\tmtwo'}$$
    \item Let $\tm=\vartwo$. Then,
      $$\tm\isub\var\tmtwo=\vartwo\todist^*\vartwo=\tm\isub\var{\tmtwo'}$$
    \item Let $\tm=\la\vartwo\tmthree$. By \ih,
      $\tmthree\isub\var\tmtwo\todist^*\tmthree\isub\var{\tmtwo'}$.
      Then,
      $$\tm\isub\var\tmtwo=\la\vartwo\tmthree\isub\var\tmtwo\todist^*\la\vartwo\tmthree\isub\var{\tmtwo'}=\tm\isub\var{\tmtwo'}$$
    \item Let $\tm=\tmthree\tmfour$. By \ih,
      $\tmthree\isub\var\tmtwo\todist^*\tmthree\isub\var{\tmtwo'}$
      and
      $\tmfour\isub\var\tmtwo\todist^*\tmfour\isub\var{\tmtwo'}$.
      Then,
      $$\tm\isub\var\tmtwo=(\tmthree\isub\var\tmtwo)(\tmfour\isub\var\tmtwo)
      \todist^*(\tmthree\isub\var{\tmtwo'})(\tmfour\isub\var{\tmtwo'})=\tm\isub\var{\tmtwo'}$$
    \item Let $\tm=\pair{\tmthree_1}{\tmthree_2}$. By \ih,
      for $i=1,2$,
      $\tmthree_i\isub\var\tmtwo \todist^* \tmthree_i\isub\var{\tmtwo'}$.
      Then,
      $$\tm\isub\var\tmtwo = \pair{\tmthree_1\isub\var\tmtwo}{\tmthree_2\isub\var\tmtwo}
      \todist^*\pair{\tmthree_1\isub\var{\tmtwo'}}{\tmthree_2\isub\var{\tmtwo'}}=\tm\isub\var{\tmtwo'}$$
    \item Let $\tm=\pi_i\tmthree$. By \ih
      $\tmthree\isub\var\tmtwo \todist^* \tmthree\isub\var{\tmtwo'}$.
      Then,
      $$\tm\isub\var\tmtwo = \pi_i{(\tmthree\isub\var\tmtwo)}
      \todist^*\pi_i{(\tmthree\isub\var{\tmtwo'})}=\tm\isub\var{\tmtwo'}
      $$
    \end{itemize}
  \end{enumerate}
\end{proof}

\xrecap{Lemma}{Substitution}{lem:substitution}
{
  If $\Gamma,x:A\vdash \tm:B$ and $\Gamma\vdash \tmtwo:A$, then $\Gamma\vdash \tm\isub x\tmtwo:B$.
}
\begin{proof}
  By induction on the derivation of $\Gamma,x:A\vdash\tm:B$.
  \begin{itemize}
  \item Let $\Gamma,x:A\vdash x:A$ as a consequence of rule $(ax)$. 
    Then, $x\isub x\tmtwo=\tmtwo$, and we have $\Gamma\vdash\tmtwo:A$.
  \item Let $\Gamma,y:B,x:A\vdash y:B$ as a consequence of rule $(ax)$.
    Then, $y\isub x\tmtwo=y$, and by rule $(ax)$, $\Gamma,y:B\vdash y:B$.
  \item Let $\Gamma,x:A\vdash\tm:B$ as a consequence of $\Gamma,x:A\vdash\tm:C$,
    $C\equiv B$ and rule $(\equiv)$. Then, by \ih,
    $\Gamma\vdash\tm\isub x\tmtwo:C$, so, by rule $(\equiv)$,
    $\Gamma\vdash\tm\isub x\tmtwo:B$.
  \item Let $\Gamma,x:A\vdash\lambda y.\tm:B\Rightarrow C$ as a consequence of
    $\Gamma,x:A,y:B\vdash \tm:C$ and rule $(\Rightarrow_i)$. Then, by \ih, $\Gamma,y:B\vdash\tm\isub x\tmtwo:C$, so, by rule
    $(\Rightarrow_i)$, $\Gamma\vdash\lambda y.\tm\isub x\tmtwo:B\Rightarrow C$.
    Notice that $\lambda y.\tm\isub x\tmtwo = (\lambda y.\tm)\isub x\tmtwo$.
  \item Let $\Gamma,x:A\vdash tr:B$ as a consequence of $\Gamma,x:A\vdash
    t:C\Rightarrow B$, $\Gamma,x:A\vdash r:C$, and rule $(\Rightarrow_e)$.
    Then, by \ih, $\Gamma\vdash t\isub x\tmtwo:C\Rightarrow
    B$ and $\Gamma\vdash r\isub x\tmtwo:C$, so, by rule $(\Rightarrow_e)$,
    $\Gamma\vdash t\isub x\tmtwo r\isub x\tmtwo:B$. Notice that $t\isub x\tmtwo
    r\isub x\tmtwo=(tr)\isub x\tmtwo$.
  \item Let $\Gamma,x:A\vdash\pair{\tm_1}{\tm_2}:B_1\wedge B_2$ as a
    consequence of $\Gamma,x:A\vdash\tm_i:B_i$, $i=1,2$, and rule $(\wedge_i)$.
    Then, by \ih, $\Gamma\vdash\tm_i\isub x\tmtwo:B_i$, so,
    by rule $(\wedge_i)$, $\Gamma\vdash\pair{\tm_1\isub x\tmtwo}{\tm_2\isub x\tmtwo}:B_1\wedge B_2$. Notice that $\pair{\tm_1\isub x\tmtwo}{\tm_2\isub x\tmtwo}=\pair{\tm_1}{\tm_2}\isub x\tmtwo$.
  \item Let $\Gamma,x:A\vdash\pi_1\tm:B$ as a consequence of
    $\Gamma,x:A\vdash\tm:B\wedge C$ and rule $(\wedge_{e_1})$.
    Then, by \ih, $\Gamma\vdash\tm\isub x\tmtwo:B\wedge C$,
    so, by rule $(\wedge_{e_1})$, $\Gamma\vdash\pi_1(\tm\isub x\tmtwo):B$.
    Notice that $\pi_1(\tm\isub x\tmtwo)=\pi_1\tm\isub x\tmtwo$.
  \item Let $\Gamma,x:A\vdash\pi_2\tm:B$ as a consequence of
    $\Gamma,x:A\vdash\tm:B\wedge B$ and rule $(\wedge_{e_1})$. Analogous to
    previous case.
    \qedhere
  \end{itemize}
\end{proof}

\end{document}